\documentclass[a4paper, leqno, 12pt]{article}
\usepackage{amsmath,amsthm}
\usepackage{amssymb,latexsym}
\usepackage{enumerate}
\usepackage[T1]{fontenc}
\usepackage{array}
\usepackage[all]{xypic}
\usepackage{bm}

\usepackage{tikz}
\usetikzlibrary{decorations.pathmorphing}
\usetikzlibrary{shapes,decorations.markings,decorations.pathreplacing,backgrounds,positioning,calc}

\usepackage[left=3cm,right=3cm,top=3.0cm,bottom=3.0cm]{geometry}

\def\BPhi{\boldsymbol{\Phi}}

\def\CCC{{\mathbb{C}}}

\def\NN{{\mathbb{N}}}
\def\MM{{\mathbb{M}}}

\def\RR{{\mathbb{R}}}

\def\CF{{\cal F}}

\def\CH{{\cal H}}

\def\CL{{\cal L}}

\def\CS{{\cal S}}
\def\CZ{{\cal Z}}

\def\CZ{{\cal Z}}
\def\CV{{\cal V}}

\def\rem{\textnormal{rem}}
\def\sgn{\textnormal{sgn}}
\def\snull{\textnormal{ }}

\def\epv {{$\mbox{}$\hfill ${\Box}$\vspace*{1.5ex} }}

\newcommand{\ra}{\rightarrow}

\newcommand{\Ra}{\Rightarrow}
\newcommand{\Lra}{\Leftrightarrow}

\def\ov{\overline}

\def\epv {{$\mbox{}$\hfill ${\Box}$\vspace*{1.5ex} }}

\newtheorem{thm}{Theorem}[section]
\newtheorem{cor}[thm]{Corollary}

\newtheorem{prop}[thm]{Proposition}

\theoremstyle{definition}
\newtheorem{df}[thm]{Definition}
\newtheorem{alg}[thm]{Algorithm}
\newtheorem{ex}[thm]{Example}
\newtheorem{remark}[thm]{Remark}

\title{\textmd{{\Large On one-parameter families of hermiticity-preserving superoperators which are not positive}}}
\date{}
\begin{document}
\maketitle
\begin{center}
Grzegorz Pastuszak$^{a}$\footnote{Corresponding author}, Alicja Jaworska-Pastuszak$^{a}$, Takeo Kamizawa$^{b}$ and Andrzej Jamio{\l}kowski$^{c}$\\\textnormal{ }\\
\footnotesize{$^{a}$ Faculty of Mathematics and Computer Science, Nicolaus Copernicus University, Toru\'n, Poland, past@mat.umk.pl (Grzegorz Pastuszak),
jaworska@mat.umk.pl (Alicja Jaworska-Pastuszak); $^{b}$ Department of Informatics and Data Science, Sanyo-Onoda City University, Sanyo-Onoda, Japan,
res.kamizawa@gmail.com; $^{c}$ Faculty of Physics, Astronomy and Informatics, Nicolaus Copernicus University, Toru\'n, Poland, jam@fizyka.umk.pl}
\end{center}

\begin{abstract} A \emph{one-parameter family of hermiticity-preserving superoperators} is a time-dependent family $\{\Phi_{t}\colon\MM_{n}(\CCC)\ra\MM_{n}(\CCC)\}_{t\in\RR}$ of hermiticity-preserving superoperators determined, in a certain sense, by real
and complex polynomial functions in the variable $t\in\RR$. The paper studies sufficient computable criteria for nonpositivity of superoperators in one-parameter families. More precisely, we give sufficient conditions for the following assertions to hold: $(1)$ every $\Phi_{t}$ is not positive, $(2)$ $\Phi_{t}$ is not positive for $t$ in some open interval $(u,v)\subseteq\RR$ and $(3)$ there is some $\Phi_{t}$ which is not positive. We show that in some situations $(3)$ implies $(2)$. Our approach to the problem is based on the \emph{Descartes rule of signs} and the \emph{Sturm-Tarski theorem}. In order to apply these facts, we introduce the \emph{sign variation formulas}. These formulas are first order logical formulas in one free variable $t$, generalising sign sequences of polynomials used in Descartes rule of signs.
\end{abstract}


\section{Introduction and the main results}

In the algebraic approach to quantum mechanics, a quantum mechanical system is represented by a special linear map $\Phi$ defined on the $\CCC$-vector space $\MM_{n}(\CCC)$ of all $n\times n$ complex matrices. Such maps are called \textit{superoperators}. It is well known that a superoperator $\Phi:\MM_{n}(\CCC)\ra\MM_{n}(\CCC)$ has an \emph{operator-sum representation} given by $$\Phi(X)=\sum_{r=1}^{s}A_{r}XB_{r}$$ with fixed complex matrices $A_{r},B_{r}\in\MM_{n}(\CCC)$, for any $r=1,\hdots,s$. In quantum mechanics, the observables of the system are identified with \emph{hermitian operators}, that is, complex selfadjoint matrices. This implies that we are interested in superoperators which \emph{preserve hermiticity}. A superoperator $\Phi\in\CL(\MM_{n}(\CCC))$ preserves hermiticity if and only if it is given by the formula $$\Phi(X)=\sum_{r=1}^{s}\alpha_{r}A_{r}XA_{r}^{*}$$ where $\alpha_{r}$ are non-zero real numbers ($A^*$ denotes the matrix adjoint to $A$). Recall that if all $\alpha_{r}$ are positive, then $\Phi$ is \emph{completely positive}. In physical terminology, these maps are often called \textit{quantum operations} or \textit{quantum channels} and play prominent role in the theory of quantum measurements and evolution of open quantum systems. We refer the reader to \cite{BZ,HZ,NC} for more details on quantum information theory.

It is in general too restrictive to consider only completely positive maps. We refer the reader to \cite{Al,Pe1,Pe2,CM} for a famous discussion on this topic and also to \cite{JWW} for more recent considerations supporting this point of view. We would like to emphasize that in the fundamental papers \cite{SMJ,JS} the authors do not assume complete positivity for quantum systems, but rather a general dynamical framework in terms of linear maps of \emph{density operators}. Density operators represent physical states and in algebraic terms they are semipositive matrices with unital traces. Furthermore, evolution of the system is described by maps on the set of all density operators. It follows that we should be interested not only in superoperators which preserve hermiticity, but also in those that are \emph{positive}. Recall that positive maps send, by definition, semipositive operators to semipositive ones.\footnote{Completely positive maps are also positive.}

Since general form of positive maps is not known, it is then crucial to find \emph{effective methods} for checking if a given hermiticity-preserving superoperator is positive or not. This problem has already attracted much attention, see for example \cite{Ja1,Ja2,Ja3,Ch,SkZ,PSJ1,PSJ2}. In particular, in \cite{PSJ1,PSJ2}\footnote{The article \cite{PSJ2} is a letter version of \cite{PSJ1}, addressed to the physical science community.} the first and the last author of the present paper give an algorithm determining whether a hermiticity-preserving superoperator is positive. The algorithm \cite[5.7]{PSJ1} is based on logical techniques of \emph{quantifier elimination theory of real numbers} \cite{Ta,Mi,Ar,Re3}. Even though it is fully deterministic, its enormous complexity \cite[Section 6]{PSJ1} makes it hard to use in practice (we refer to \cite{BKR,BSS,BPR} for more information about complexity of quantifier elimination). Moreover, the approach taken in \cite{PSJ1} does not consider time-dependency of a quantum system.

%Moreover, the physical states are given by the \emph{density operators} which are semipositive matrices with unital traces. Evolutions of the system are described by maps on the set of all density operators. This implies that we are interested in superoperators which preserve hermiticity and such that are \emph{positive}, that is, send semipositive operators to semipositive ones. A superoperator $\Phi$ preserves hermiticity if and only if it is of the form $$\Phi(X)=\sum_{r=1}^{s}\alpha_{r}A_{r}XA_{r}^{*}$$ where $\alpha_{r}$ are non-zero real numbers and $A^*$ denotes the matrix adjoint to $A$. Recall that if all the numbers $\alpha_{r}$ are positive, then $\Phi$ is \emph{completely positive} and thus positive as well. In physical terminology these maps are often called \textit{quantum operations} or \textit{quantum channels} and play prominent role in evolution of open quantum systems and in the theory of quantum measurements. However, a general form of positive maps is not known. Since such maps describe the physical world, it is then crucial to find effective methods allowing to distinguish in the class of hermiticity-preserving superoperators those which are positive or not.

Indeed, in quantum mechanics we are often interested in \emph{time-dependent evolution} of quantum dynamical systems. Such evolution can often be described as a collection $$\BPhi=\{\Phi_{t}:\MM_{n}(\RR)\ra\MM_{n}(\RR)\}_{t\in\RR}$$ of superoperators, indexed by a real parameter $t$. If the system is isolated, the time evolution generates a one-parameter
group $\BPhi=\{\Phi_{t}\}_{t\in\RR}$ of unitary operators where $\Phi_{t}=e^{-iHt}$ for a given time-independent Hamiltonian $H$. On the other hand, if the system $\BPhi$ is open, thus interacts with environment, its evolution is given by  $$\Phi_{t}(X)=\sum_{r=1}^{s}A_{r}(t)XA_{r}(t)^{*},$$ for some complex time-dependent matrices $A_{r}(t)$ such that $\sum_{r=1}^{s}A_{r}(t)^{*}A_{r}(t)$ equals the identity matrix, for any parameter $t\in\RR$. Note that $\Phi_{t}$ are still completely positive in this case, although of more complicated form than $e^{-iHt}$.

Based on the discussion above we are convinced to study a general situation in which $\Phi_{t}$ are only assumed to preserve hermiticity, and so may or may not be positive. It follows that $\Phi_{t}$ takes the form: $$\Phi_{t}(X)=\sum_{r=1}^{s}\alpha_{r}(t)A_{r}(t)XA_{r}(t)^{*}$$ where $\alpha_{r}(t)$ are non-zero real numbers, for any parameter $t\in\RR$. We call such systems \emph{general one-parameter families of hermiticity-preserving superoperators}. This paper initiates a systematic study of the problem of positivity and nonpositivity of superoperators in these families. In a series of articles, we aim to give \emph{computable criteria} for determining which superoperators in one-parameter families are positive or nonpositive. By computable criterion (or \emph{effective condition}), we mean
a procedure employing only finite number of arithmetic operations. In other words, a computable criterion is an \emph{algorithm} employing arithmetic operations on real or complex numbers. We refer the reader to \cite{JP1,JP2,Pa} for important examples of such criteria in the context of quantum channels.

In the present paper we focus on general one-parameter families of a special form. Namely, we assume that the functions $\alpha_{r}$ are real polynomials in variable $t$ (without real roots) and the matrices $A_{r}$ have complex polynomial functions as coefficients, see Definition 4.1 for the precise statement. We call these families \textit{one-parameter families of hermiticity-preserving superoperators} (without the adjective \textit{general}). The aim of the paper is to give sufficient computable criteria for existence of nonpositive superoperators in one-parameter families. Here we are concerned with three main cases: $(1)$ every $\Phi_{t}$ is not positive, $(2)$ $\Phi_{t}$ is not positive for $t$ in some open interval $(u,v)\subseteq\RR$ and $(3)$ there is some $\Phi_{t}$ which is not positive. A one-parameter family is called \emph{globally nonpositive} in the first case, \emph{locally nonpositive} in the case of $(2)$ and \emph{pointwise nonpositive} in the last case, see Definition 4.2.

Our methods are based on two constructive mathematical tools: the \emph{Descartes rule of signs} and the \emph{Sturm-Tarski theorem}. The Descartes rule of signs (see Theorem 3.1) relates the number of positive real roots of a real univariate polynomial with the number of \emph{sign variations} in the sequence of signs associated with coefficients of the polynomial. This sequence is called the \textit{sign sequence} of a univariate polynomial. The Sturm-Tarski theorem (see Theorem 6.2) allows us to calculate the number of distinct real roots of real polynomials under some additional conditions (see Theorem 6.3 and Proposition 6.5). In \cite{Ja3} the author applies the Descartes rule to give a handy criterion for nonpositivity of a fixed hermiticity-preserving superoperator (recalled in Corollary 3.3). Here we adapt the results from \cite{Ja3} to the case of one-parameter families. For this purpose, we introduce the \textit{sign variation formulas} which are generalizations of sign sequences for bivariate polynomials. It turns out that global, local and pointwise nonpositivity of a one-parameter family can be checked by looking at conditions of the form $\forall_{t\in\RR\snull}\varphi(t)$, $\exists_{t\in(u,v)\snull}\varphi(t)$ and $\exists_{t\in\RR\snull}\varphi(t)$, respectively, where $\varphi(t)$ is a sign variation formula of a special form (see Theorem 5.5). Using the Sturm-Tarski theorem we prove that these conditions are computable. This fact is shown in Theorem 6.7 which is the first main result of the paper. The second main result is Theorem 7.1, presenting the computable criterion for nonpositivity of one-parameter families.

The paper is organized as follows. Sections 2 and 3 are introductory. In Section 2 we recall basic information about superoperators, especially the Choi-Jamio{\l}kowski isomorphism which plays important role in our considerations. In Proposition 2.4 we recall the known fact that a positive hermiticity-preserving superoperator $\Phi\in\CL(\MM_{n}(\CCC))$ has number of negative eigenvalues bounded by $(n-1)^2$. This fact is fundamental for the criterion of \cite{Ja3} and serves as a basis for our results as well. Section 3 is devoted to Descartes rule of signs. In particular, we recall in Corollary 3.3 the main result of \cite{Ja3}. In Section 4 we introduce the one-parameter families of hermiticity-preserving superoperators, see Definition 4.1. We are interested in determining whether a given one-parameter family is globally, locally or pointwise nonpositive (Definition 4.2). The main result of the section is Proposition 4.3 which shows that the set of all characteristic polynomials of a one-parameter family may be viewed as a bivariate real polynomial in $\RR[t,x]$. This result is a consequence of the fact that one-parameter families are determined by polynomials, in a suitable sense. However, we believe that they may be nevertheless good approximations of more general situations (see Remark 4.5). Section 5 introduces the sign variation formulas. We show in Theorem 5.5 how these formulas can be applied in the
verification of global, local and pointwise nonpositivity of one-parameter families. Moreover, Theorem 5.7 states that local and pointwise nonpositivity are equivalent conditions, if considered in terms of sign variation formulas. Sections 6 and 7 are the core of the paper. Indeed, the purpose of Section 6 is to give computable criteria for conditions $\exists_{t\in\RR\snull}\varphi(t)$ and $\forall_{t\in\RR\snull}\varphi(t)$ to hold, see Theorem 6.3 and Corollary 6.6, respectively. Both facts are merged for convenience in Theorem 7.1. These criteria are based on Sturm-Tarski theorem and hence a part of the section is devoted to recalling some related facts and definitions. Section 7 is the final one. In this section we present the main computable criterion for the global, local and pointwise nonpositivity of a one-parameter family. The criterion is given in Theorem 7.1 which is also a basis for two algorithms (Algorithms 7.2 and 7.3). The section ends with an application of Algorithm 7.2.

\section{Hermiticity-preserving superoperators}

This section recalls some basic facts (and fixes related notation) regarding \emph{superoperators}. In particular, we recall in Theorem 2.1 the
\emph{Choi-Jamio{\l}kowski isomorphism} which relates positivity of a superoperator $\Phi$ on $\MM_{n}(\CCC)$ with block positivity of the corresponding linear map $\mathcal{J}(\Phi)$ on $\CCC^{n}\otimes\CCC^{n}$. In Proposition 2.2 we calculate the matrix form $M(\mathcal{J}(\Phi))$ of $\mathcal{J}(\Phi)$. We finish the section with Proposition 2.4, relating the number of negative eigenvalues of $M(\mathcal{J}(\Phi))$ with nonpositivity of hermiticity-preserving superoperator $\Phi$.

In the paper we consider only finite dimensional vector spaces over the field $\CCC$ of complex numbers. If $V$ is such a vector space, then
$\mathcal{L}(V)$ denotes the $\mathbb{C}$-vector space of all linear maps $T\colon V\to V$. Recall that if the dimension of $\mathcal{L}(V)$ equals $n$, then $\mathcal{L}(V)$ is isomorphic with the space $\MM_{n}(\CCC)$ of all $n\times n$ complex matrices.\footnote{The same holds if $V$ is a Hilbert space of dimension $n$.} We do not distinguish between elements of $\mathcal{L}(V)$ and $\MM_{n}(\CCC)$. We usually consider vector spaces such as $\CCC^{n},\CCC^{n}\otimes\CCC^{n},\MM_{n}(\CCC)$. It is useful to recall
that the isomorphism $\CL(\CCC^{n})\cong\MM_{n}(\CCC)$ of $\mathbb{C}$-vector spaces yields
$$\CL(\CCC^{n}\otimes\CCC^{n})\cong\CL(\CCC^{n})\otimes\CL(\CCC^{n})\cong\MM_{n}(\CCC)\otimes\MM_{n}(\CCC).$$ An element of $\CL(\MM_{n}(\CCC))$ is called a \textit{superoperator}.

We denote by $e_1,\dots,e_n$ the elements of the standard $\CCC$-basis of $\CCC^{n}$. The simple tensors $e_i\otimes e_j$, for $i,j=1,\dots,n$, form
the standard $\CCC$-basis of the vector space $\CCC^n\otimes\CCC^n$ and are denoted by $\epsilon_{ij}$. The $n\times n$ complex matrices
$E_{ij}=[e_{kl}]$ such that $e_{kl}\in\{0,1\}$ and $e_{kl}=1 \Lra (k,l)=(i,j)$ form the standard $\CCC$-basis of $\MM_{n}(\CCC)$.

It is well known that $\mathbb{C}^n$ is a Hilbert space with respect to the standard inner product $\langle\cdot\!\mid\!\cdot\rangle$ such that for any
$x=(x_{1},\dots,x_{n}),y=(y_{1},\dots,y_{n})\in\CCC^{n}$ we have $$\langle  x\!\mid\!  y \rangle=\sum_{i=1}^{n}\ov{x_{i}}y_{i}$$ where $\ov{x_{i}}$
denotes the complex conjugate of $x_{i}$. The space $\CCC^{n}\otimes\CCC^{n}$ is also a Hilbert space with respect to the inner product defined as
$$\langle x\otimes y\!\mid\! x'\otimes y'\rangle=\langle x\!\mid\! x'\rangle\cdot\langle y\!\mid\! y'\rangle,$$ for any $x,x',y,y'\in\CCC^{n}$.

Assume that $\CH$ is a finite dimensional Hilbert space with the inner product $\langle \cdot\!\mid\!\cdot\rangle$. If $T\in\mathcal{L}(\CH)$, then there
exists a unique \textit{adjoint operator} $T^*\!\in\mathcal{L}(\CH)$ such that, for any
$x,y\in\CH$, the equality $\langle Tx\!\mid\! y \rangle=\langle x\!\mid\! T^*y\rangle$ holds. Note that if $T\in\MM_{n}(\CCC)$, then $T^*$ is the conjugate transpose of $T$, that is, the matrix \textit{adjoint} to $T$.

An operator $T\in\CL(\CH)$ is \emph{selfadjoint} (or \emph{hermitian}) if and only if $T=T^{*}$ which is in turn equivalent to the condition $\langle x\!\mid\! Tx\rangle\in\mathbb{R}$, for any $x\in\CH$. A selfadjoint operator $T\in\CL(\CH)$ is \textit{semipositive} if and only if $\langle x\!\mid\! Tx\rangle\geq 0$, for any $x\in\CH$.

Assume that $\Phi\in\CL(\MM_{n}(\CCC))$ is a superoperator. We say that $\Phi$ \textit{preserves hermiticity} if and only if $\Phi(X)$ is hermitian, whenever $X\in\MM_{n}(\CCC)$ is hermitian. It is well known that this condition holds if $\Phi$ is given by the following formula:
$$\Phi(X)=\sum_{r=1}^{s}\alpha_{r}A_{r}XA_{r}^{*}$$ where $s\in \NN$ and $\alpha_{r}\in\RR$ are non-zero and $A_r, A_r^*\in\MM_{n}(\CCC)$, for any $r=1,\dots,s$. If case $\alpha_{r}>0$, then $\Phi$ is known to be \textit{completely positive}. This condition is equivalent with assuming that $\alpha_{r}=1$, for any $r=1,\dots,s$.

We call $\Phi$ \emph{positive} if and only if $\Phi(X)$ is semipositive for any semipositive operator $X\in\MM_{n}(\CCC)$. It is well known that completely positive maps are positive. An operator $T\in\CL(\CCC^n\otimes\CCC^n)$ is called \emph{block positive} if and only if $\langle x\otimes y\!\mid\! T(x\otimes y)\rangle\geq 0$, for any $x,y\in\CCC^{n}$. We recall the following renowned result, see for example \cite{Ja1,Ja2,Ch1}.

\begin{thm} The map \[\mathcal{J}\colon\CL(\MM_n(\CCC))\to\MM_n(\CCC)\otimes\MM_n(\CCC)\cong\CL(\CCC^{n}\otimes\CCC^{n})\] defined by the formula \[\mathcal{J}(\Phi)=\sum_{i,j=1}^{n}E_{ij}\otimes\Phi(E_{ij})\] is an isomorphism of Hilbert spaces. Moreover, a superoperator $\Phi$ is positive if and only if $\mathcal{J}(\Phi)$ is block positive, for any $\Phi\in\CL(\MM_n(\CCC))$.
\end{thm} The isomorphism $\mathcal{J}$ from the above theorem is known as the \textit{Choi-Jamio{\l}kowski isomorphism}.

Assume that $\Phi\in\CL(\MM_n(\CCC))$ preserves hermiticity. Observe that in this case $\mathcal{J}(\Phi)\in\MM_n(\CCC)\otimes\MM_n(\CCC)$ is selfadjoint. Indeed, since $\Phi(X)^{*}=\Phi(X^{*})$, we get $$\mathcal{J}(\Phi)^{*}=(\sum_{k,l=1}^{n}E_{kl}\otimes\Phi(E_{kl}))^{*}=\sum_{k,l=1}^{n}E_{kl}^{*}\otimes\Phi(E_{kl})^{*}=
\sum_{k,l=1}^{n}E_{lk}\otimes\Phi(E_{lk})=\mathcal{J}(\Phi),$$ because $E_{kl}^{*}=E_{lk}$, for any $k,l=1,\dots n$.
Now we calculate the matrix of the operator $\mathcal{J}(\Phi)$ in the standard basis $\epsilon_{ij}=e_{i}\otimes e_{j}$ of the space $\CCC^{n}\otimes\CCC^{n}$. We call this matrix the \emph{$CJ$-matrix} for $\Phi$ and denote it by $M(\mathcal{J}(\Phi))$.

Assume that $T\in\CL(\CCC^{n}\otimes\CCC^{n})$ is an operator. We denote by $T_{(ij)(kl)}$ the complex numbers such that $$T(\epsilon_{ij})=\sum_{k,l=1}^{n}T_{(ij)(kl)}\epsilon_{kl},$$ for any $i,j=1,\hdots,n$. The following fact is proved in \cite[Theorem 3.5 (2)]{PSJ1}, see also \cite[Section 2]{PSJ2}. We sketch its proof for convenience.

\begin{prop} Assume that $\Phi\in\CL(\MM_{n}(\CCC))$ is a hermiticity-preserving superoperator such that $\Phi(X)=\sum_{r=1}^{s}\alpha_{r}A_{r}XA_{r}^{*}$ and $A_{r}=[a_{ij}^{r}]$, for any $r=1,\dots,s$. Then we have the equality $$\mathcal{J}(\Phi)_{(ij)(kl)}=\sum_{r=1}^{s}\alpha_{r}a_{lk}^{r}\ov{a_{ji}^{r}},$$ for any $i,j,k,l=1,\dots,n$.
\end{prop}
\begin{proof} Since obviously $$\mathcal{J}(\Phi)(\epsilon_{ij})=\sum_{k,l=1}^{n}(E_{kl}e_{i})\otimes(\Phi(E_{kl})e_{j})=\sum_{k=1}^{n}e_{k}\otimes(\Phi(E_{ki})e_{j}),$$ we conclude by straightforward calculations that $$\mathcal{J}(\Phi)_{(ij)(kl)}=e_{l}^{tr}\Phi(E_{ki})e_{j}=\sum_{r=1}^{s}\alpha_{r}a_{lk}^{r}\ov{a_{ji}^{r}}.$$ This shows the claim.
\end{proof} Observe that $\ov{\mathcal{J}(\Phi)_{(ij)(kl)}}=\mathcal{J}(\Phi)_{(kl)(ij)}$, for any $i,j,k,l\in 1,\dots,n$, which also shows that $\mathcal{J}(\Phi)$ is selfadjoint. Consequently, the $CJ$-matrix $M(\mathcal{J}(\Phi))$, for $\Phi$ defined as above, is an $n^{2}\times n^{2}$ complex selfadjoint matrix which is the transpose of the matrix $[m_{(ij)(kl)}]$ whose rows and columns are indexed by pairs $(i,j)$ and $(k,l)$ (where $i,j,k,l=1,\hdots,n$) such that $$m_{(ij)(kl)}=\sum_{r=1}^{s}\alpha_{r}a_{lk}^{r}\ov{a_{ji}^{r}}.$$
Moreover, since $\mathcal{J}(\Phi)_{(ij)(kl)}=e_{l}^{tr}\Phi(E_{ki})e_{j}$, we get that $M(\mathcal{J}(\Phi))$ is a block matrix of the form
\[
M\left(\mathcal{J}\left(\Phi\right)\right)=\left[\begin{array}{ccc}
\Phi(E_{11})  & \cdots & \Phi(E_{1n})\\
\vdots  &  & \vdots\\
\Phi(E_{n1})  & \cdots &\Phi(E_{nn})
\end{array}\right].
\] We denote by $\chi_{\Phi}$ the characteristic polynomial of $M(\mathcal{J}(\Phi))$, that is \[\chi_{\Phi}(x)=\det(M(\mathcal{J}(\Phi))-xI)\] where $I$ is the $n^{2}\times n^{2}$ identity matrix.

\begin{remark} Since the matrix $M(\mathcal{J}(\Phi))\in\MM_{n^{2}}(\CCC)$ is selfadjoint, the characteristic polynomial $\chi_{\Phi}$ is real (of degree $n^{2}$) and has only real roots, meaning that any complex root of $\chi_{\Phi}$ is in fact real. In other words, the $CJ$-matrix $M(\mathcal{J}(\Phi))$ has only real eigenvalues.
\end{remark}

It is known that if $\CH_{1},\CH_{2}$ are Hilbert spaces with dimensions $m,n$, respectively, then any block positive operator on $\CH_{1}\otimes\CH_{2}$ has the number of negative eigenvalues (counted with multiplicities) not greater then $(m-1)(n-1)$, see for example \cite[Section 3]{Ja3}. This yields the following property which is a fundamental ingredient of our main results.

\begin{prop} Assume that $\Phi\in\CL(\MM_n(\CCC))$ is a hermiticity-preserving superoperator. If the number of all negative roots of $\chi_{\Phi}$, counted with multiplicities, is greater or equal to $(n-1)^{2}+1=n^{2}-2n+2$, then $\Phi$ is not positive.
\end{prop}

\section{Descartes rule of signs}

In its most general form, the \emph{Descartes rule of signs} states that the number of positive roots of a univariate real polynomial equals the number of \emph{sign variations} in its \emph{sign sequence}, minus some even number, see Theorem {\ref{3.2}} (1). Moreover, if the polynomial has only real roots, then both numbers coincide, see Theorem \ref{3.2} (2). It was observed in \cite{Ja3} that this fact, together with Proposition 2.4, gives a handy criterion for the nonpositivity of hermiticity-preserving maps. The criterion is recalled in Corollary 3.3. 

We start this section by establishing the notation and terminology which are necessary to express the Descartes rule of signs in a proper way. Assume that $\Sigma=\{+,-\}$ is an alphabet and denote by $\Sigma^{*}$ the set of all words over $\Sigma$. We view the set $\Sigma^{*}$ as a monoid whose binary operation is the concatenation of words and neutral element is the empty word. If $\sigma,\tau\in\Sigma^{*}$, then $\sigma\tau$ denotes the concatenation of $\sigma$ and $\tau$. If $n\in\NN$, then $\sigma^{n}$ denotes the word $\sigma\sigma\hdots\sigma$ ($n$ is the number of words $\sigma$ in this sequence). If $\sigma\in\Sigma^{*}$, then $l(\sigma)$ is the \emph{length} of $\sigma$, understood in a natural way. If $i=1,\hdots,l(\sigma)$, then $\sigma(i)\in\Sigma$ is the $i$-th letter of $\sigma$, counted from the right.
A subword of $\sigma\in\Sigma^{*}$ of the form $+-$ or $-+$ is called a \emph{sign variation} in $\sigma$. We denote by $\lambda(\sigma)$ the
number of sign variations in $\sigma$. Moreover, with any sequence $\underline{a}=(a_{n},\hdots,a_{0})\in\RR^{n+1}$ of non-zero real numbers we associate a word $s(\underline{a})=s(a_{n})s(a_{n-1})\hdots s(a_{0})\in\Sigma^{*}$, where $s(a)=+$ if $a>0$, otherwise $s(a)=-$,
for any non-zero $a\in \mathbb{R}$.

In the paper we use the standard notation  $\RR[x]$, $\CCC[x]$, $\RR[t,x]$ and $\CCC[t,x]$ for the rings of polynomials in one variable $x$ or two variables $t,x$ over fields $\RR$ and $\CCC$. Polynomials are identified with polynomial functions. Elements of $\RR[t,x]$ and $\CCC[t,x]$ are sometimes called \emph{bivariate polynomials}. In this paper we often treat bivariate polynomials as univariate polynomials over rings of polynomials, i.e. we identify $\RR[t,x]$ with $\RR[t][x]$ and $\CCC[t,x]$ with $\CCC[t][x]$. For $f \in \CCC[t,x]$ and $t_0$ a fixed complex number we denote by $f(t_0)$ the unique polynomial in $\CCC[x]$ such that $f(t_{0})(x)=f(t_{0},x)$, for any $x\in\CCC$.

Assume now that $f\in\RR[x]$ and $f(x)=a_{n_{k}}x^{n_{k}}+a_{n_{k-1}}x^{n_{k-1}}+\hdots + a_{n_{0}}x^{n_{0}}$, where $n_{k}>n_{k-1}>\hdots>n_{0}$ is a strictly decreasing sequence of nonzero natural numbers. Obviously, every polynomial $f$ can be written in that form. We call $\underline{a}=(a_{n_{k}},a_{n_{k-1}},\hdots,a_{n_{0}})$ the \emph{coefficient sequence} of $f$. We set $$s(f)=s(a_{n_{k}},a_{n_{k-1}},\hdots,a_{n_{0}})\in\Sigma^{*}$$ and call it the \textit{sign sequence} of the polynomial $f$. We also denote by $\varrho(f)$, $\varrho_{+}(f)$ and $\varrho_{-}(f)$ the number of all real roots of $f$, the number of all positive roots of $f$ and the number of all negative roots of $f$, respectively. In any case we count the number of roots together with multiplicities.

The following useful fact is known as the \textit{Descartes rule of signs}, see for example \cite[Section 2]{BPR}. In the paper we apply only the second assertion of this theorem.

\begin{thm} \label{3.2} Assume that $f\in\RR[x]$. The following assertions hold.
\begin{enumerate}[\rm(1)]
  \item We have $\varrho_{+}(f)=\lambda(s(f))-2k$ where $k\geq 0$ is some natural number, i.e. the number of positive roots of $f$ equals
      the number of sign variations in the sign sequence $s(f)$ of $f$ minus some even natural number.
  \item  Assume that every complex root of $f\in\RR[x]$ is real. In this case $\varrho_{+}(f)=\lambda(s(f))$, i.e. the number of positive roots of $f$
      equals the number of sign variations in the sign sequence $s(f)$ of $f$.
\end{enumerate}
\end{thm}

Let us illustrate some of the above notions in a simple example.
\begin{ex} Assume that $f(x)=x^{5}+3x^{4}-x^{3}-7x^{2}+4=(x-1)^{2}(x+2)^{2}(x+1).$ Then $(1,3,-1,-7,4)$ is the coefficient sequence of $f$ and $$s(f)=s(1,3,-1,-7,4)=++--+$$ is its sign sequence. Moreover, we have $\varrho(f)=5$, $\varrho_{+}(f)=2$ and $\varrho_{-}(f)=3$. Observe that $\lambda(s(f))=2$ which is consistent with Theorem \ref{3.2} (2), because $f$ has only real roots.
\end{ex}

The following useful criterion for nonpositivity of hermiticity-preserving superoperators, based on the crucial Proposition 2.4, was first shown in \cite{Ja3}. Our aim is to generalize this criterion for \emph{one-parameter families} of hermiticity-preserving superoperators which we introduce in the next section.

For $f\in\RR[x]$ we denote by $f^{-}\in\RR[x]$ a polynomial such that $f^{-}(x)=f(-x)$, for any $x\in\RR$. Obviously, $r$ is a negative (positive, respectively) root of $f$ if and only if $-r$ is a positive (negative, respectively) root of $f^{-}$.

\begin{cor} Assume that $\Phi\in\CL(\MM_{n}(\CCC))$ preserves hermiticity and $\chi_{\Phi}\in\RR[x]$ is the characteristic polynomial of the $CJ$-matrix $M(\mathcal{J}(\Phi))$. If $\lambda(s(\chi_{\Phi}^{-}))\geq n^{2}-2n+2$, then $\Phi$ is not positive.
\end{cor}

\begin{proof} It follows from Remark 2.3 that $\chi_{\Phi}^{-}$ has only real roots and hence Theorem \ref{3.2} (2) implies that $\varrho_{+}(\chi_{\Phi}^{-})=\lambda(s(\chi_{\Phi}^{-}))$. Since $\varrho_{+}(\chi_{\Phi}^{-})$ is exactly the number of negative roots of $\chi_{\Phi}$ (counted with multiplicities), the assertion is a consequence of Proposition 2.4.
\end{proof}

\section{One-parameter families of superoperators}

This section introduces the \emph{one-parameter families of hermiticity-preserving superoperators}. These families are sets of hermiticity-preserving
superoperators, indexed by one real parameter $t$, which are determined (in the sense of Definition 4.1) by real and complex polynomial
functions in the variable $t$. We would like to know if such one-parameter families contain hermiticity-preserving superoperators which are not
positive (see Definition 4.2).

In Proposition 4.3 we show that the set of all characteristic polynomials of a one-parameter family of hermiticity-preserving
superoperators may be viewed as a bivariate polynomial in $\RR[t][x]$. This fact has important consequences in next sections, as does Corollary
4.4 which is a natural generalisation of Corollary 3.3 for the case of one-parameter families. The final remark of the section presents an argument  for one parameter families of hermiticity-preserving superoperators to be good approximations of more general situations.

We introduce some notation. Let $\CCC[t_{\RR}]$ denote the set of all polynomial functions $f\colon\RR\to\CCC$ where $f\in\CCC[t]$. For simplicity,
elements of $\CCC[t_{\RR}]$ are also called polynomials. If $f\in\CCC[t_{\RR}]$ and $f(t)=a_{n}t^{n}+\hdots+a_{0}$, then we define
$\ov{f}\in\CCC[t_{\RR}]$ as \[\ov{f}(t):=\ov{f(t)}=\ov{a_{n}t^{n}+\hdots+a_{0}}=\ov{a_{n}}t^{n}+\hdots+\ov{a_{0}}\] and call $\ov{f}$ the
\emph{conjugate} of $f$. Denote by $\MM_{n}(\CCC[t_{\RR}])$ the set of all $n\times n$ matrices with coefficients in $\CCC[t_{\RR}]$. Elements of
$\MM_{n}(\CCC[t_{\RR}])$ are called \textit{one-parameter matrices}.

Assume that $A\in\MM_{n}(\CCC[t_{\RR}])$ is a one-parameter matrix such that $A=[a_{ij}]$. We define the one-parameter matrix
$A^{*}\in\MM_{n}(\CCC[t_{\RR}])$ as the conjugate transpose of $A$, that is, $A^{*}=[\ov{a_{ji}}]$. For a given $t\in\RR$, $A^{*}(t)$ is the matrix adjoint to $A(t)$.

\begin{df} A \emph{one-parameter family of hermiticity-preserving superoperators} is a function $\BPhi\colon\RR\ra\CL(\MM_{n}(\CCC))$ such that there exist a natural number $s\geq1$, real polynomials $\alpha_{1},\hdots,\alpha_{s}\in\RR[t]$ with no real roots and one-parameter matrices $A_{1},\hdots,A_{s}\in\MM_{n}(\CCC[t_{\RR}])$ such that $$\BPhi(t)(X)=\sum_{r=1}^{s}\alpha_{r}(t)A_{r}(t)XA_{r}^{*}(t),$$ for any $t\in\RR$ and $X\in\MM_{n}(\CCC)$.
\end{df} We usually denote a one-parameter family $\BPhi\colon\RR\ra\CL(\MM_{n}(\CCC))$ of hermiticity-preserving superoperators as $$\BPhi=\{\Phi_{t}\colon\MM_{n}(\CCC)\ra\MM_{n}(\CCC)\}_{t\in\RR}$$ where $\Phi_{t}(X):=\BPhi(t)(X)$, for any $t\in\RR$ and $X\in\MM_{n}(\CCC)$. We sometimes call $\BPhi$ a \emph{one-parameter family}, in short.

Observe that if $\BPhi$ is a one-parameter family, then any operator $\Phi_{t}$ preserves hermiticity, because the polynomials
$\alpha_{1},\hdots,\alpha_{s}\in\RR[t]$ have no real roots and hence are strictly positive or strictly negative.

\begin{df} A one-parameter family of hermiticity-preserving superoperators $\BPhi=\{\Phi_{t}\colon\MM_{n}(\CCC)\ra\MM_{n}(\CCC)\}_{t\in\RR}$ is called:
\begin{enumerate}[\rm(1)]
  \item \textit{globally nonpositive} if and only if $\Phi_{t}$ is not positive, for any $t\in\RR$.
  \item \textit{locally nonpositive} if and only if there is an open interval $(u,v)\subseteq\RR$ such that $\Phi_{t}$ is not positive, for any $t\in (u,v)$.
  \item \textit{pointwise nonpositive} if and only if there exists $t\in\RR$ such that $\Phi_{t}$ is not positive.
\end{enumerate}
\end{df}

Our aim in the paper is to apply Corollary 3.3 in order to give sufficient computable criteria for a one-parameter family to be globally nonpositive,
locally nonpositive or pointwise nonpositive.

Assume that $\BPhi=\{\Phi_{t}\colon\MM_{n}(\CCC)\ra\MM_{n}(\CCC)\}_{t\in\RR}$ is a one-parameter family of hermiticity-preserving superoperators. We
set $\chi_{t}:=\chi_{\Phi_{t}}\in\RR[x]$, for any $t\in\RR$.

The following crucial fact shows that the family $\{\chi_{t}\mid t\in\RR\}$ of all characteristic polynomials of the one-parameter family $\BPhi$ can be treated as a real polynomial in two variables $t$ and $x$.

\begin{prop} Assume that $\BPhi=\{\Phi_{t}\colon\MM_{n}(\CCC)\ra\MM_{n}(\CCC)\}_{t\in\RR}$ is a one-parameter family of hermiticity-preserving superoperators such that $$\Phi_{t}(X)=\sum_{r=1}^{s}\alpha_{r}(t)A_{r}(t)XA_{r}^{*}(t)$$ and $A_{r}(t)=[a_{ij}^{r}(t)]$. Then there is a polynomial $\chi_{\BPhi}\in\RR[t,x]$ such that $\chi_{\BPhi}(t)=\chi_{t}$, for any $t\in\RR$.
\end{prop}

\begin{proof} We conclude from Proposition 2.2 that $$\mathcal{J}(\Phi_{t})_{(ij)(kl)}=\sum_{r=1}^{s}\alpha_{r}(t)a_{lk}^{r}(t)\ov{a_{ji}^{r}(t)}=\sum_{r=1}^{s}\alpha_{r}(t)a_{lk}^{r}(t)\ov{a_{ji}^{r}}(t)$$ and hence $M(\mathcal{J}(\Phi_{t}))\in\MM_{n^{2}}(\CCC[t_{\RR}])$ is a one-parameter matrix.
This implies that the characteristic polynomial $\chi_t$ of $M(\mathcal{J}(\Phi_{t}))$ may be treated as an element of $\CCC[t][x]$, if we view $t$ as a variable. Let us denote this bivariate polynomial by $\chi_{\BPhi}$  with coefficients $h_{1},\hdots,h_{n^{2}}\in\CCC[t]$.

Observe that for a given $t_{0}\in\RR$, $\chi_{\BPhi}(t_{0})=\chi_{t_0}$ is an element of $\RR[x]$, because $M(\mathcal{J}(\Phi_{t_{0}}))$ is selfadjoint, see Remark 2.3. Hence we conclude that any coefficient $h_{i}$ of $\chi_{\BPhi}$ has the property that $h_{i}(t)\in\RR$, for any $t\in\RR$. Then it follows that $h_i\in \RR[t]$. Indeed, assume that $h\in\CCC[t]$ is a polynomial defined as $h=a_{s}t^{s}+\hdots+a_{1}t+a_{0}$ and $h(t)\in\RR$, for any $t\in\RR$. Then we get \[a_{s}t^{s}+\hdots+a_{1}t+a_{0}=h(t)=\ov{h(t)}=\ov{a_{s}}t^{s}+\hdots+\ov{a_{1}}t+\ov{a_{0}}\] and so $$(\ov{a_{s}}-a_{s})t^{s}+\hdots+(\ov{a_{1}}-a_{1})t+(\ov{a_{0}}-a_{0})=0,$$ for any $t\in\RR$. Since any nonzero complex polynomial has only finite number of roots, we conclude that $\ov{a_{s}}-a_{s}=\dots=\ov{a_{1}}-a_{1}=\ov{a_{0}}-a_{0}=0$ and thus $\ov{a_{i}}=a_{i}$, so all $a_i \in \RR$ and hence $h\in\RR[t]$.

The above arguments directly imply that $\chi_{\BPhi}\in\RR[t][x]=\RR[t,x]$.
\end{proof}

We call the bivariate polynomial $\chi_{\BPhi}\in\RR[t][x]$ the \textit{characteristic polynomial} of the one-parameter family
$\BPhi=\{\Phi_{t}\colon\MM_{n}(\CCC)\ra\MM_{n}(\CCC)\}_{t\in\RR}$. We denote by $\chi_{\BPhi}^{-}\in\RR[t][x]$ the bivariate polynomial such that
$\chi_{\BPhi}^{-}(t,x)=\chi_{\BPhi}(t,-x)$, for any $t,x\in\RR$. It is convenient to observe that, in our notation, for any $t\in\RR$ we get the following equalities:
\[\chi_{\BPhi}(t)=\chi_{t}=\chi_{\Phi_{t}}\textnormal{ and thus }\chi_{\BPhi}^{-}(t)=\chi^{-}_{t}=\chi^{-}_{\Phi_{t}}.\]

The following fact is a version of Corollary 3.3 for one-parameter families.

\begin{cor}  Assume that $\BPhi=\{\Phi_{t}\colon\MM_{n}(\CCC)\ra\MM_{n}(\CCC)\}_{t\in\RR}$ is a one-parameter family of hermiticity-preserving superoperators. The following statements hold.
\begin{enumerate}[\rm(1)]
  \item If $\lambda(s(\chi_{t}^{-}))\geq n^{2}-2n+2$, for any $t\in\RR$, then $\BPhi$ is globally nonpositive.
  \item If there is an open interval $(u,v)\subseteq\RR$ such that $\lambda(s(\chi_{t}^{-}))\geq n^{2}-2n+2$, for any $t\in(u,v)$, then $\BPhi$ is locally nonpositive.
  \item If there exists $t\in\RR$ such that $\lambda(s(\chi_{t}^{-}))\geq n^{2}-2n+2$, then $\BPhi$ is pointwise nonpositive.
\end{enumerate}
\end{cor}

\begin{proof} For any $t\in\RR$, $\chi_{t}$ is the characteristic polynomial of a hermiticity-preserving superoperator $\Phi_{t}$.
Hence the above implications follow at once from Corollary 3.3.
\end{proof}

In the above corollary we consider the number of sign variations in the sign sequence of univariate polynomials $\chi_{\BPhi}^{-}(t)=\chi^{-}_{t}\in\RR[x]$. Recall that we define in Section 3 the coefficient sequence only for a polynomial in $\RR[x]$, but the definition generalized naturally to the ring $\RR[t,x]$. Namely, for $f \in \RR[t,x]=\RR[t][x]$ and $f(t,x)=a_{n_k}(t)x^{n_k}+\ldots +a_{n_1}(t)x^{n_1}+a_{n_0}(t)x^{n_0}$ the coefficient sequence is a sequence $\underline{a}(t)=(a_{n_k}(t),\ldots,a_{n_1}(t),a_{n_0}(t))$ of real polynomials in variable $t$. Therefore we may consider the coefficient sequence of characteristic polynomial $\chi_{\BPhi}$ of the family $\BPhi$.

To adopt the concepts of sign sequence and sign variation for the case of $\underline{a}(t)$, we introduce in the next section the notion of \emph{sign variation formula}. It turns out that if $\varphi(t)$ is a sign variation formula (in one free variable $t$), then conditions $\forall_{t\in\RR\snull}\varphi(t)$ and $\exists_{t\in\RR\snull}\varphi(t)$ can be checked effectively using the \emph{Sturm-Tarski theorem} and its corollaries, see Section 6. Moreover, under some natural assumptions on a sign variation formula $\varphi(t)$ associated with $\chi_{\BPhi}^{-}$,
we have the implication \[\varphi(t)\Ra\lambda(s(\chi_{t}^{-}))\geq n^{2}-2n+2,\] for any $t\in\RR$. Therefore the language of sign variation formulas, together with the above Corollary 4.4, provides us with effective methods for verifying the nonpositivity of one-parameter families of hermiticity-preserving superoperators.

\begin{remark} In view of Definition 4.1, one-parameter families of hermiticity-preserving superoperators are determined, in suitable sense, by polynomials. This assumption is fundamental in our considerations, because it allows us applying the Sturm-Tarski theorem and its corollaries in order to obtain computable conditions for nonpositivity. Although it may seem restrictive to accept only polynomials, the restriction
can be omitted, at least in some situations. Indeed, for many applications it would be sufficient to assume that the functions in the definition of one-parameter family are Taylor expandable. Then such general one-parameter family could be approximated by series of one-parameter families from Definition 4.1. We believe that our results could be appropriately generalized to that setting. This problem is left for further research.
\end{remark}

\section{Sign variation formulas}

As mentioned before, in this section we introduce \textit{sign variation formulas}, see Definitions 5.1 and 5.3. These formulas are useful analogues of sign sequences of univariate polynomials, generalized to the setting of bivariate polynomials, see Remark 5.4. Moreover, Theorem 5.5 shows exactly how sign variation formulas are applied in the verification of global, local and pointwise nonpositivity of one-parameter families. Finally, we prove in Theorem 5.7 (see also Proposition 5.6) that, at least in the context of sign variation formulas, local
and pointwise nonpositivity are equivalent conditions.

We stress that this section heavily uses the notation introduced in Section 3, especially for the monoid $\Sigma^{*}$ of words over the alphabet $\Sigma=\{+,-\}$ (but also for polynomials).

\begin{df} A \emph{sign variation formula of type $\sigma\in\Sigma^{*}$} is a formula $\varphi(t)$ (in one free variable $t$, see \cite{Rot}) recursively defined in the following way:
\begin{itemize}
  \item if $\sigma=+$ ($\sigma=-$, respectively), then $\varphi(t)$ is of the form $a(t)>0$ ($a(t)<0$, respectively), for some $a\in\RR[t]$,
  \item if $\sigma=\tau+$ ($\sigma=\tau-$, respectively), for some $\tau\in\Sigma^{*}$, then $\varphi(t)$ is of the form
      $\phi(t)\wedge a(t)>0$ ($\phi(t)\wedge a(t)<0$, respectively), for some sign variation formula $\phi(t)$ of type
      $\tau$ and $a\in\RR[t]$.
\end{itemize}
\end{df}
A \emph{sign variation formula} is a sign variation formula $\varphi(t)$ of some type $\sigma\in\Sigma^{*}$. If we want to emphasize the type of a sign variation formula, then we denote $\varphi(t)$ as $\varphi(t)_{\sigma}$, if the type of $\varphi(t)$ equals $\sigma$.

Observe that a sign variation formula $\varphi(t)$ of type $\sigma\in\Sigma^{*}$ is in fact a formula of the form
$$\bigwedge_{i=0}^{k}a_{k-i}(t)\Delta_{k-i}0$$ where $a_{i}(t)\in\RR[t]$, $\Delta_{i}\in\{<,>\}$, and $\Delta_{i}$ equals $>$ if and
only if $\sigma(i)=+$, for any $i=0,\hdots,k$. (it follows that $\Delta_{i}$ equals $<$ if and only if $\sigma(i)=-$) We call any $a_{i}(t)>0$ a
\emph{positive component} of $\varphi(t)$ and $a_{j}(t)<0$ a \emph{negative component} of $\varphi(t)$. The sequence of polynomials
$(a_{k},a_{k-1},\hdots,a_{0})$ is called the \emph{domain} of $\varphi(t)$.

\begin{ex} Consider the formula $\varphi(t)$ of the form $$a_{5}(t)>0\;\wedge\; a_{4}(t)>0\;\wedge\; a_{3}(t)<0\;\wedge\; a_{2}(t)>0\;\wedge\; a_{1}(t)<0\;\wedge\; a_{0}(t)>0$$ where $a_{5},\hdots,a_{0}\in\RR[t]$ are some polynomials. Then $\varphi(t)$ is a sign variation formula of the type $\sigma=++-+-+$, i.e. we have $\varphi(t)=\varphi(t)_{\sigma}$. Moreover, formulas $a_{5}(t)>0$, $a_{4}(t)>0$, $a_{2}(t)>0$, $a_{0}(t)>0$ are positive components of $\varphi(t)_{\sigma}$ whereas $a_{3}(t)<0$ and $a_1(t)<0$ are its only two negative components. The domain of $\varphi(t)_{\sigma}$ is the sequence $(a_{5},a_{4},\hdots,a_{0})$.

\end{ex} It is easy to see that any sign variation formula is equivalent with a sign variation formula of the type $+^{n}$. Indeed, if $\varphi(t)$ is a sign variation formula, then we define the \emph{normal form} of $\varphi(t)$ as the formula obtained from $\varphi(t)$ by replacing any negative component $a(t)<0$ of $\varphi(t)$ by the formula $(-a(t))>0$. Observe that the normal form of the formula from Example 5.2 looks as follows: \[a_{5}(t)>0\;\wedge\; a_{4}(t)>0\;\wedge\; -a_{3}(t)>0\;\wedge\; a_{2}(t)>0\;\wedge\; -a_{1}(t)>0\;\wedge\; a_{0}(t)>0.\]
It is obvious that a sign variation formula is equivalent with its normal form. In Section 6 we shall assume that any sign variation formula is written in the normal form.

\begin{df} Assume that $f\in\RR[t][x]$ is a bivariate polynomial such that $$f(t,x)=a_{n_{k}}(t)x^{n_{k}}+a_{n_{k-1}}(t)x^{n_{k-1}}+\hdots+a_{n_{0}}(t)x^{n_{0}},$$ for some strictly decreasing sequence $n_{k}>n_{k-1}>\hdots>n_{0}$ of natural numbers and non-zero polynomials $a_{n_{i}}\in\RR[t]$.
Let $\sigma$ be a word from $\Sigma^*$ of length $k+1$.
A \emph{sign variation formula of type $\sigma$ for $f$} is a sign variation formula $\varphi(t)_{\sigma}$ whose domain is the coefficient sequence $(a_{n_{k}},a_{n_{k-1}},\hdots,a_{n_{0}})$ of the polynomial $f$. We denote this formula by $\varphi(t)_{\sigma}^{f}$.
\end{df}

The crucial property of a sign variation formula of type $\sigma$ for a bivariate polynomial $f\in\RR[t][x]$ lies in the following straightforward
observation.

\begin{remark} Assume that $\varphi(t)$ is a sign variation formula of type $\sigma$ for a polynomial $f\in\RR[t][x]$, i.e. $\varphi(t)=\varphi(t)_{\sigma}^{f}$. Fix a real number $t_0\in\RR$. Then $\varphi(t_0)_{\sigma}^{f}$ holds if and only if the sign sequence of the polynomial $f(t_0)\in\RR[x]$ equals $\sigma$, i.e. $s(f(t_0))=\sigma$. Hence in this situation we have two equalities:
$$s(f(t_0))=\sigma\textnormal{ and thus }\lambda(s(f(t_0)))=\lambda(\sigma).$$
\end{remark}

The following result is a version of Corollary 4.4, written in terms of sign variation formulas. This fact shows that sign variation formulas
provide a convenient language for verification of nonpositivity of one-parameter families.

\begin{thm} Assume that $\BPhi=\{\Phi_{t}\colon\MM_{n}(\CCC)\ra\MM_{n}(\CCC)\}_{t\in\RR}$ is a one-parameter family of hermiticity-preserving superoperators and let $f=\chi_{\BPhi}^{-}\in\RR[t][x]$. Assume that $\varphi(t)_{\sigma}^{f}$ is a sign variation formula such that $\lambda(\sigma)\geq n^{2}-2n+2$. Then the following assertions hold.
\begin{enumerate}[\rm(1)]
  \item If the formula $\forall_{t\in\RR\snull}\varphi(t)_{\sigma}^{f}$ holds, then $\BPhi$ is globally nonpositive.
  \item If there is an open interval $(u,v)\subseteq\RR$ such that the formula $\forall_{t\in(u,v)\snull}\varphi(t)_{\sigma}^{f}$ holds, then $\BPhi$ is locally nonpositive.
  \item If the formula $\exists_{t\in\RR\snull}\varphi(t)_{\sigma}^{f}$ holds, then $\BPhi$ is pointwise nonpositive.
\end{enumerate}
\end{thm}

\begin{proof} Assume that $t_{0}$ is a fixed real number. We conclude from the assumptions and Remark 5.4 that if $\varphi(t_0)_{\sigma}^{f}$ holds, then we get $$\lambda(s(\chi_{t_0}^{-}))=\lambda(s(f(t_0)))=\lambda(\sigma)\geq n^{2}-2n+2.$$ Hence all assertions are consequences of Corollary 4.4.\footnote{Note that $\lambda(s(f(t_0)))=\lambda(\sigma)$ does not imply $s(f(t_0))=\sigma$, so the condition $\lambda(s(f(t_0)))=\lambda(\sigma)\geq n^{2}-2n+2$ does not imply that $\varphi(t_0)_{\sigma}^{f}$ holds.}
\end{proof}

The fact that the characteristic polynomial $\chi_{\BPhi}$ of the one-parameter family $\BPhi$ is a bivariate real polynomial, proved in Proposition
4.3, is an implicit but crucial element of the above theorem. It should be emphasized that the idea of sign variation formulas arose from that very
source.

We say that a sign variation formula $\varphi(t)$ is \emph{satisfied} if and only if there is some $t_0\in\RR$ such that $\varphi(t_0)$ holds.
Equivalently, $\varphi(t)$ is satisfied if and only if the formula $\exists_{t\in\RR\snull}\varphi(t)$ holds.

\begin{prop} A sign variation formula $\varphi(t)$ is satisfied if and only if there is an open interval $(u,v)\subseteq\RR$, $u<v$, such that $\varphi(t)$ holds for any $t\in(u,v)$.
\end{prop}

\begin{proof} We only show one implication since the converse is obvious. Assume that $\varphi(t)$ is satisfied and let $\varphi(t_0)$ hold, for some $t_0\in\RR$. It follows that $a_{i}(t_0)\Delta_{i} 0$ holds for any component $a_{i}(t)\Delta_{i} 0$ of $\varphi(t)$. Since $a_{i}$ is a real polynomial, we conclude that $a_{i}(t)\Delta_{i} 0$ holds for any $t$ belonging to some interval $(u_{i},v_{i})$ which contains $t_0$. Then the intersection $\bigcap_{i=0}^k(u_{i},v_{i})$ of these intervals is an non-empty interval (containing $t_0$) satisfying the required condition.
\end{proof}

We obtain the following refinement of Theorem 5.5 (3).

\begin{thm} Assume that $\BPhi=\{\Phi_{t}\colon\MM_{n}(\CCC)\ra\MM_{n}(\CCC)\}_{t\in\RR}$ is a one-parameter family of hermiticity-preserving superoperators and let $f=\chi_{\BPhi}^{-}\in\RR[t][x]$. Assume that $\varphi(t)_{\sigma}^{f}$ is a sign variation formula such that $\lambda(\sigma)\geq n^{2}-2n+2$. If the formula $\exists_{t\in\RR\snull}\varphi(t)_{\sigma}^{f}$ holds, then $\BPhi$ is locally nonpositive.
\end{thm}

\begin{proof} The assertion follows from Proposition 5.6 and arguments similar as in the proof of Theorem 5.5.
\end{proof}

\section{Applying Sturm-Tarski theorem}

The purpose of this section is to give computable criteria for the sentences $\exists_{t\in\RR\snull}\varphi(t)$ and
$\forall_{t\in\RR\snull}\varphi(t)$ to hold, if $\varphi(t)$ is a sign variation formula. The criterion for $\forall_{t\in\RR\snull}\varphi(t)$
is given in Theorem 6.3, whereas for $\exists_{t\in\RR\snull}\varphi(t)$ in Corollary 6.6. Both theorems are presented together for convenience in
Theorem 6.7, the first main result of the paper. Our criteria are based on the \emph{Sturm-Tarski theorem} (also known as \emph{generalized Sturm theorem}), recalled in Theorem 6.2, and some of its corollaries, see especially Proposition 6.5.

Without loss of generality we assume that $\varphi(t)$ is always written in the normal form. For simplicity, we write $\forall_{t\snull}\varphi(t)$ and $\exists_{t\snull}\varphi(t)$ instead of $\forall_{t\in\RR\snull}\varphi(t)$ and $\exists_{t\in\RR\snull}\varphi(t)$.
We start with introducing some notation and terminology associated with the Sturm-Tarski theorem. We base on \cite{XY}.

\begin{df}
Assume that $n\in\NN$. A tuple $(h_{0},h_{1},\hdots,h_{n})\in\RR[x]^{n+1}$ of non-zero polynomials is a \textit{Sturm sequence} if and only if the following conditions hold:
\begin{enumerate}[\rm(1)]
	\item The polynomial $h_{n}$ does not have real roots.
	\item If $h_{i}(x)=0$, for some $x\in\RR$, then $h_{i-1}(x)h_{i+1}(x)<0$, for any $i=1,\dots,n-1$.
\end{enumerate}
\end{df}

We recall now the canonical recursive construction of a Sturm sequence associated with two non-zero polynomials $p,q\in\RR[x]$. Up to the sign, its elements are polynomials that occur as remainders in the Euclid's algorithm for determining the greatest common divisor of $p$ and $q$. Namely, let
$h_{0}=p$ and $h_{1}=q$. Assume that $n\geq 1$ and $h_{0},h_{1},\hdots,h_{n}$ are defined. If $h_{n}\!\mid \! h_{n-1}$, then
we have a Sturm sequence of the form $(h_{0},h_{1},\hdots,h_{n})$. Otherwise, we set $h_{n+1}=-\rem_{h_{n}}(h_{n-1})$ where $\rem_{h_{n}}(h_{n-1})$ is the remainder of division of polynomial $h_{n-1}$ by $h_{n}$. The sequence obtained in this construction is called \emph{the canonical Sturm sequence} for polynomials $p$ and $q$.

Note that if $(h_{0},h_{1},\hdots,h_{n})$ is the canonical Sturm sequence for $p$ and $q$, then it follows from the Euclid's algorithm that $h_{n}$ is, up to the sign, the greatest common divisor of $p$ and $q$ and thus $h_{n}\! \mid \! h_{i}$, for any $i=0,\hdots,n$. It is well known that the sequence $$\left(\frac{h_{0}}{h_{n}},\frac{h_{1}}{h_{n}},\hdots,\frac{h_{n}}{h_{n}}=1\right)$$ is also a Sturm sequence, see \cite{XY}. Other examples of Sturm sequences can be found in \cite{BPR}.

Let now $h\in\RR[x]$ and $u=-\infty$ or $u=\infty$. If $\lim_{x\ra u}h(x)=\infty$, then we set $\sigma_{u}(h)=+$. Otherwise, we set $\sigma_{u}(h)=-$. If $h_{1},\hdots,h_{s}\in\RR[x]$, then we set \[\sigma_{u}(h_{1},\hdots,h_{s})=(\sigma_{u}(h_{1}),\hdots,\sigma_{u}(h_{s}))\in\Sigma^{*}.\]

Assume that $p,q\in\RR[x]$ are non-zero polynomials and let $(h_{0},h_{1},\hdots,h_{n})$ be the canonical Sturm sequence for $p$ and $q$. We define $\nu(p,q)$ as the number \[\lambda(\sigma_{-\infty}(h_{0},\hdots,h_{n}))-\lambda(\sigma_{\infty}(h_{0},\hdots,h_{n})).\]
Moreover, we define $N(p,q)$ as the number $$|\{x\in\RR\mid p(x)=0\wedge q(x)>0\}|-|\{x\in\RR\mid p(x)=0\wedge q(x)<0\}|$$ where $|X|$ denotes the number of elements of a set $X$. Observe that if $q$ is a polynomial such that $q(x)>0$, for any $x\in\RR$, then $N(p,q)$ is the number of distinct real roots of polynomial $p$.

The following theorem is known as the \textit{generalized Sturm theorem} or the \textit{Sturm-Tarski theorem}. We denote by $\textbf{1}$ the constant polynomial $q(x)=1$.

\begin{thm} Assume that $p,q\in\RR[x]$ are non-zero polynomials. Then the following equality holds: $$\nu(p,p'q)=N(p,q)$$ and thus the number $N(p,q)$ can be computed. In particular, the number $N(p,\textnormal{\textbf{1}})$ of all distinct real roots of $p$ can be computed as $\nu(p,p')$. \epv
\end{thm}

We refer the reader to \cite{XY} and \cite{BPR} for more general versions of the above theorem. In this section we apply only the above special case.

In the sequel we assume that a sign variation formula is given in its normal form. This means that we are interested in determining the validity of
the sentences $\exists_{t\snull}\varphi(t)$ and $\forall_{t\snull}\varphi(t)$ where $\varphi(t)=a_{k}(t)>0\wedge\hdots\wedge a_{0}(t)>0$,
for some $a_{i}(t)\in\RR[t]$. For the case of the sentence $\forall_{t\snull}\varphi(t)$ this is given
in the following theorem.

\begin{thm} Assume that $\varphi(t)=a_{k}(t)>0\wedge\hdots\wedge a_{0}(t)>0$ is a sign variation formula in its normal form. Then the sentence $\forall_{t\snull}\varphi(t)$ holds if and only if the following condition is satisfied: $$\bigwedge_{i=0}^{k}N(a_{i}(t),\textnormal{\textbf{1}})=0\wedge a_{i}(0)>0.$$
\end{thm}

\begin{proof} The sentence $\forall_{t\snull}\varphi(t)$ holds if and only if the sentence $\forall_{t\snull}a_{i}(t)>0$ holds, for any $i=1,\dots,n$. Hence it is sufficient to show that if $a(t)\in\RR[t]$, then the sentence $\forall_{t\snull}a(t)>0$ holds if and only if the condition $N(a(t),\textnormal{\textbf{1}})=0\wedge a(0)>0$ is satisfied. The condition $N(a(t),\textnormal{\textbf{1}})=0$ means that the polynomial $a(t)$ has no real roots (by Sturm-Tarski theorem) and hence its graph lies entirely above or below the abscissa axis (which follows from the Darboux property of $\RR$). The additional condition that $a(0)>0$ ensures that the former holds. Thus we get $\forall_{t\snull}a(t)>0$ which finishes the proof of one implication. Since $\forall_{t\snull}a(t)>0$ yields $a$ has no roots and $a(0)>0$, the other implication holds as well. This finishes the proof.
\end{proof}

A computable condition for the second formula $\exists_{t\snull}\varphi(t)$ is related with counting the number of elements of the set $$\CS(f,a_{k},\hdots,a_{0}):=\{t\in\RR\mid f(t)=0\;\wedge\; a_{k}(t)>0\;\wedge\;\hdots\;\wedge\; a_{0}(t)>0\}$$ where $f,a_{i}\in\RR[t]$ are non-zero polynomials. We show in the proposition below how this can be done.

Recall that the \emph{sign function} $\sgn\colon\RR\ra\{-1,0,1\}$ is defined as $\sgn(r)=-1$, if $r<0$, $\sgn(r)=0$, if $r=0$ and $\sgn(r)=1$, if
$r>0$.

\begin{remark} It is easy to see that $\sgn(r_{1}r_{2})=\sgn(r_{1})\sgn(r_{2})$, for any $r_{1},r_{2}\in\RR$. Moreover, for any $f,g\in\RR[x]$ we have $$N(f,g)=\sum_{x\in\CZ(f)}\sgn(g(x))$$ where $\CZ(f)$ is the set of all distinct real roots of $f$.
\end{remark}

\begin{prop} Assume that $f,a_{k},\hdots,a_{0}\in\RR[t]$ are non-zero polynomials. Then we have the equality $$|\CS(f,a_{k},\hdots,a_{0})|=\frac{1}{2^{k+1}}\sum_{(p_{0},\hdots,p_{k})\in\{1,2\}^{k+1}}N(f,a_{k}^{p_{k}}\cdot a_{k-1}^{p_{k-1}}\cdot\hdots\cdot a_{0}^{p_{0}}).$$
\end{prop}

\begin{proof} Define $l(r)=\frac{1}{2}(r^{2}+r)$ and observe that $l(1)=1$ and $$l(r)=0\Lra r=0\vee r=-1.$$ Consequently, for any polynomial $h\in\RR[x]$ and a fixed $x_0\in \RR$ we have $h(x_0)>0$ if and only if $l(\sgn(h(x_0)))\neq 0$, and in this case $l(\sgn(h(x_0)))=1$. It follows that $$|\CS(f,a_{k},\hdots,a_{0})|=\sum_{t\in\CZ(f)}l(\sgn(a_{k}(t)))\cdot l(\sgn(a_{k-1}(t)))\cdot\hdots\cdot l(\sgn(a_{0}(t))).$$
Moreover, Remark 6.4 yields the following equalities: \[\sum_{t\in\CZ(f)}l(\sgn(a_{k}(t)))\cdot l(\sgn(a_{k-1}(t)))\cdot\hdots\cdot l(\sgn(a_{0}(t)))=\]
\[=\sum_{t\in\CZ(f)}\prod_{i=0}^{k}\frac{1}{2}(\sgn(a_{i}^{2}(t))+\sgn(a_{i}(t)))=\]
$$=\sum_{t\in\CZ(f)}\frac{1}{2^{k+1}}\left( \sum_{(p_{k},\hdots,p_{0})\in\{1,2\}^{k+1}}\sgn(a_{k}^{p_{k}}(t))\cdot\sgn(a_{k-1}^{p_{k-1}}(t))\cdot
\hdots\cdot\sgn(a_{0}^{p_{0}}(t))\right)=$$ $$=\frac{1}{2^{k+1}}\sum_{(p_{k},\dots,p_{0})\in\{1,2\}^{k+1}}\sum_{t\in\CZ(f)}
\sgn(a_{k}^{p_{k}}(t)\cdot a_{k-1}^{p_{k-1}}(t)\cdot\hdots\cdot a_{0}^{p_{0}}(t))=$$
$$=\frac{1}{2^{k+1}}\sum_{(p_{k},\hdots,p_{0})\in\{1,2\}^{k+1}}N(f,a_{k}^{p_{k}}\cdot a_{k-1}^{p_{k-1}}\cdot\hdots\cdot a_{0}^{p_{0}}).$$ This shows the claim.
\end{proof}

The above result generalizes \cite[Proposition 4.2]{PSJ1}, see also \cite[Section 3]{PSJ2}. Similarly as before, the proof applies some ideas from Section 5 of \cite{KB}.

We denote by $\ov{\RR}$ the set $\RR\cup\{-\infty,\infty\}$ and we assume that $-\infty<u<\infty$, for any $u\in\RR$. In the following fact we show that the above proposition can be applied to give a computable criterion for checking whether the sentence $\exists_{t\snull}\varphi(t)$ holds or not. This is a generalisation of \cite[Corollary 4.3]{PSJ1} (see also \cite[Section 3]{PSJ2}).

\begin{cor} Assume that $\varphi(t)=a_{k}(t)>0\wedge\hdots\wedge a_{0}(t)>0$ is a sign variation formula in its normal form.
\begin{enumerate}[\rm(1)]
  \item The following two conditions are computable:
  \begin{enumerate}[\rm(a)]
    \item $\lim_{t\ra\infty}a_{i}(t)=\infty$, for all $i=0,\dots,k$,
    \item $\lim_{t\ra-\infty}a_{i}(t)=\infty$, for all $i=0,\hdots,k$.
  \end{enumerate} If either of them is satisfied, then the formula $\exists_{t\snull}\varphi(t)$ holds.
  \item Assume that $\lim_{t\ra\infty}a_{l}(t)\neq\infty$ and $\lim_{t\ra-\infty}a_{m}(t)\neq\infty$, for some $l,m=0,\hdots,k$. Then the formula $\exists_{t\snull}\varphi(t)$ is equivalent with the formula $$\exists_{t\snull}a_{k}(t)>0\wedge\hdots\wedge a_{0}(t)>0\wedge(a_{k}\cdot\hdots\cdot a_{0})^{'}(t)=0.$$ Therefore these formulas are equivalent with the condition $$|\CS((a_{k}\cdot\hdots\cdot a_{0})^{'},a_{k},\hdots,a_{0})|\neq 0.$$
\end{enumerate}

\end{cor}
\begin{proof} Observe that conditions on limits from both statements complement each other.

$(1)$ It is easy to see that both conditions can be checked by looking at signs of leading coefficients and degrees of polynomials $a_{k},\hdots,a_{0}$. The latter implication is obvious.

$(2)$ Clearly it is sufficient to show that if there is $t_0\in\RR$ such that the condition $a_{k}(t_0)>0\wedge\hdots\wedge a_{0}(t_0)>0$ holds, then there is $t_1\in\RR$ such that $$(*)\quad a_{k}(t_1)>0\wedge\hdots\wedge a_{0}(t_1)>0\wedge(a_{k}\cdot\hdots\cdot a_{0})'(t_1)=0.$$ Hence assume the former and let $I_{j}=(u_{j},v_{j})\subseteq\ov{\RR}$ be the largest interval such that $t_0\in I_{j}$ and $a_{j}(t)>0$, for any $t\in I_{j}$ and $j=0,\hdots k$. Observe that if $u_{j}$ or $v_{j}$ is finite, then it is a root of $a_{j}$. Let $I=\bigcap_{j=0}^{k}I_j$ and note that $I$ is an interval containing $t_{0}$. Since $\lim_{t\ra\infty}a_{l}(t)\neq\infty$ and $\lim_{t\ra-\infty}a_{m}(t)\neq\infty$ we obtain that $v_{l}\neq\infty$ and $u_{m}\neq-\infty$, respectively. Hence $I=(u,v)$ where $u$ and $v$ are roots of some polynomials from the set $\{a_{k},\hdots,a_{0}\}$ and thus roots of the polynomial $a=a_{k}\cdot\hdots\cdot a_{0}$. We get $a(t)>0$, for any $t \in I$, because $a_{j}(t)>0$, for any $t\in I$, $j=0,\hdots,k$. Then Rolle's theorem implies that there is $t_1\in I$ with $a'(t_1)=0$. This finishes the proof of $(*)$, and hence $(2)$ follows by Proposition 6.5.
\end{proof}
We summarize our results in the following main theorem.

\begin{thm} Assume that $\varphi(t)=a_{k}(t)>0\wedge\hdots\wedge a_{0}(t)>0$ is a sign variation formula in its normal form. The following statements hold.
\begin{enumerate}[\rm(1)]
  \item If the computable condition $$\bigwedge_{i=0}^{k}N(a_{i}(t),\textnormal{\textbf{1}})=0\wedge a_{i}(0)>0$$ holds, then the formula $\forall_{t\snull}\varphi(t)$ holds.
  \item If one of the following computable conditions hold:
  \begin{enumerate}
    \item[\rm(2.1)] $\lim_{t\ra\infty}a_{i}(t)=\infty$, for all $i=0,\dots,k$,
    \item[\rm(2.2)] $\lim_{t\ra-\infty}a_{i}(t)=\infty$, for all $i=0,\dots,k$,
    \item[\rm(2.3)] $|\CS((a_{k}\cdot\hdots\cdot a_{0})^{'},a_{k},\hdots,a_{0})|\neq 0$, but $(2.1)$ and $(2.2)$ do not hold,
  \end{enumerate} then the formula $\exists_{t\snull}\varphi(t)$ holds.
\end{enumerate}
\end{thm}

\begin{proof} The assertion of $(1)$ follows from Theorem 6.3 and that of $(2)$ from Corollary 6.6 and Proposition 6.5.
\end{proof}

\section{The criterion and some examples}

The final section presents the computable criterion for nonpositivity of one-parameter families of hermiticity-preserving superoperators. The
criterion is given in Theorem 7.1 which is the second main result of the paper (after Theorem 6.7). It builds directly on results of Section 5 and
Section 6. Based on this theorem, we develop two algorithms (Algorithm 7.2 and Algorithm 7.3) for determining the nonpositivity of one-parameter families. We finish the section with an example of application of Algorithm 7.3.

\begin{thm} Assume that $\BPhi=\{\Phi_{t}\colon\MM_{n}(\CCC)\ra\MM_{n}(\CCC)\}_{t\in\RR}$ is a one-parameter family of hermiticity-preserving superoperators and let $f=\chi_{\BPhi}^{-}\in\RR[t][x]$. Assume that $\varphi(t)_{\sigma}^{f}$ is a sign variation formula such that $\lambda(\sigma)\geq n^{2}-2n+2$ and let $\varphi(t)=a_{k}(t)>0\wedge\hdots\wedge a_{0}(t)>0$ be its normal form. The following assertions hold.
\begin{enumerate}[\rm(1)]
  \item If the computable condition $$\bigwedge_{i=0}^{k}N(a_{i}(t),\textnormal{\textbf{1}})=0\wedge a_{i}(0)>0$$ holds, then the formula $\BPhi$ is globally nonpositive.
  \item If one of the following computable conditions hold:
  \begin{enumerate}
    \item[\rm(2.1)] $\lim_{t\ra\infty}a_{i}(t)=\infty$, for all $i=0,\dots,k$,
    \item[\rm(2.2)] $\lim_{t\ra-\infty}a_{i}(t)=\infty$, for all $i=0,\dots,k$,
    \item[\rm(2.3)] $|\CS((a_{k}\cdot\hdots\cdot a_{0})',a_{k},\hdots,a_{0})|\neq 0$, but $(2.1)$ and $(2.2)$ do not hold,
  \end{enumerate} then $\BPhi$ is pointwise nonpositive and locally nonpositive.
\end{enumerate}
\end{thm}

\begin{proof} The assertions follow directly from Theorem 5.5, 5.7 and 6.7.
\end{proof}
For a one-parameter family  $\BPhi=\{\Phi_{t}\colon\MM_{n}(\CCC)\ra\MM_{n}(\CCC)\}_{t\in\RR}$ of hermiticity-preserving superoperators we
define $\CS\CV\CF(\BPhi)$ as the following set of sign variation formulas: $$\CS\CV\CF(\BPhi):=\{\varphi(t)_{\sigma}^{f}\mid
f=\chi^{-}_{\BPhi}\textnormal{ and }\lambda(\sigma)\geq n^{2}-2n+2\}.$$
The algorithm below checks whether a fixed sign variation formula $\varphi(t)\in\CS\CV\CF(\BPhi)$ gives any information about the nonpositivity of $\BPhi$. In the algorithm we assume that the characteristic
polynomial $\chi_{\BPhi}\in\RR[t][x]$ of $\BPhi$ is already calculated (and so is the polynomial $\chi^{-}_{\BPhi}$).

\begin{alg}\textnormal{ }\\
Input:
\begin{itemize}
  \item a one-parameter family $\BPhi$ of hermiticity-preserving superoperators,
  \item the characteristic polynomial $\chi_{\BPhi}\in\RR[t][x]$ of $\BPhi$,
  \item a sign variation formula $\varphi(t)\in\CS\CV\CF(\BPhi)$.
\end{itemize}
Output:
\begin{itemize}
  \item \textit{$\BPhi$ is globally nonpositive},
  \item \textit{$\BPhi$ is locally nonpositive and pointwise nonpositive},
  \item \textit{The formula $\varphi(t)$ does not determine nonpositivity of $\BPhi$}.
\end{itemize}
Steps:
\begin{enumerate}[\bf(S1)]
  \item Determine the normal form of the sign variation formula $\varphi(t)$. Assume that $$a_{k}(t)>0\wedge\hdots\wedge a_{0}(t)>0$$ is
      the normal form of $\varphi(t)$.
  \item Determine whether the condition $$\bigwedge_{i=0}^{k}N(a_{i}(t),\textnormal{\textbf{1}})=0\wedge a_{i}(0)>0$$ holds. If so, then finish with the output \[\textit{$\BPhi$ is globally nonpositive}.\] Otherwise go to Step (S3).
  \item Determine whether
  \begin{itemize}
    \item $\lim_{x\ra\infty}a_{i}(x)=\infty$, for any $i=0,\dots,k$, or
    \item $\lim_{x\ra-\infty}a_{i}(x)=\infty$, for any $i=0,\dots,k$.
  \end{itemize} If any of these conditions holds, then finish with the output \[\textit{$\BPhi$ is locally nonpositive and pointwise nonpositive}.\] Otherwise go to Step (S4).
  \item Determine whether the condition $$|\CS((a_{k}\cdot\hdots\cdot a_{0})',a_{k},\hdots,a_{0})|\neq 0$$ holds.
   If so, then finish with the output \[\textit{$\BPhi$ is locally nonpositive and pointwise nonpositive}.\] Otherwise go to Step (S5).
  \item Finish with the output \[\textit{The formula $\varphi(t)$ does not determine nonpositivity of $\BPhi$}.\]
\end{enumerate}
\end{alg}

\begin{proof} Correctness of the algorithm follows from Theorem 7.1.
\end{proof}

The second algorithm shows how to apply the previous one to any sign variation formula from the set $\CS\CV\CF(\BPhi)$. Generally speaking, the
algorithm shows the range of application of Theorems 6.7 and 7.1. More specifically, it answers the question to what extent Descartes rule of signs, Sturm-Tarski theorem and the language of sign variation formulas can be used to verify of the nonpositivity of one-parameter families.

\begin{alg} \textnormal{ }\\
Input:
\begin{itemize}
  \item \textit{ a one-parameter family $\BPhi$ of hermiticity-preserving superoperators.}
\end{itemize}
Output:
\begin{itemize}
  \item \textit{$\BPhi$ is globally nonpositive},
  \item \textit{$\BPhi$ is locally nonpositive and pointwise nonpositive},
  \item \textit{$\CS\CV\CF(\BPhi)$ does not determine nonpositivity of $\BPhi$}.
\end{itemize}
Steps:
 \begin{enumerate}[\bf(S1)]
    \item Calculate the characteristic polynomial $\chi_{\BPhi}\in\RR[t][x]$ of $\BPhi$ and determine the set $\CS\CV\CF(\BPhi)$.
    \item Set $\CS\CV\CF(\BPhi)_0:=\CS\CV\CF(\BPhi)$.
    \item Take a sign variation formula $\varphi(t)\in\CS\CV\CF(\BPhi)_0$.
    \item Apply Algorithm 7.2 to the formula $\varphi(t)$. If its output is different than \[\textit{The formula $\varphi(t)$ does not determine nonpositivity of $\BPhi$},\] then finish. Otherwise, go to Step (S5).
    \item Check if $\CS\CV\CF(\BPhi)_0\neq\emptyset$. If so, then set $$\CS\CV\CF(\BPhi)_0:=\CS\CV\CF(\BPhi)_0\setminus\{\varphi(t)\}$$ and go back to Step (S3). Otherwise finish with the output \[\textit{$\CS\CV\CF(\BPhi)$ does not determine nonpositivity of $\BPhi$.}\]
  \end{enumerate}
\end{alg}
\begin{proof} Correctness of the algorithm follows from Theorem 7.1.
\end{proof}

The section is finished with an example of application of Algorithm 7.2. We purposely avoided long and tedious calculations which are not essential in a paper of theoretical nature. In particular, we assume that one-parameter matrices occurring in the considered one-parameter family are $2\times 2$ real matrices. This makes the calculations easy to comprehend. Indeed, our humble goal is only to give a glimpse of the computational methods presented in this section. Observe however that both algorithms are straightforward to implement in any high level computer algebra system.

\begin{ex} Consider a one-parameter family $\BPhi=\{\Phi_{t}:\MM_{2}(\RR)\ra\MM_{2}(\RR)\}_{t\in\RR}$ such that
\[
\Phi_{t}\left(X\right)=\sum_{k=1}^{3}\alpha_{k}A_{k}\left(t\right)XA^{*}_{k}\left(t\right)
\] where $\alpha_{1}=\alpha_{2}=\alpha_{3}=-1$ and
\begin{align*}
A_{1}\left(t\right) & =\left[\begin{array}{cc}
0 & -t\\
0 & 1
\end{array}\right],\;A_{2}\left(t\right)=\left[\begin{array}{cc}
1 & 0\\
-\!1\!-\!t & -t
\end{array}\right],
A_{3}\left(t\right) =\left[\begin{array}{cc}
0 & 0\\
1 & 1\!-\!t
\end{array}\right].
\end{align*} Then we obtain
\[
M\left(\mathcal{J}\left(\Phi_{t}\right)\right)=\left[ \begin {array}{cccc} -1&t+1&0&t\\\noalign{\smallskip}t+1& -t^2-2t-2 & 0 & -t^2-1\\
\noalign{\smallskip}0&0&-{t}^{2}&t\\
\noalign{\smallskip}t& -t^2-1 & t & -2{t}^{2}+2t-2 \end {array} \right],
\]
hence
\begin{align*}
\chi_{\BPhi}(x)=x^4\!+\!(4t^2\!+\!5)x^3\!
+\!(4t^4\!+\!2t^3\!+\!5t^2\!-\!2t\!+\!6)x^2\!+\!(t^6\!+\!2t^5\!\!-\!4t^3\!+\!3t^2\!+\! 1)x,
\end{align*}
and consequently
\begin{align*}
\chi_{\BPhi}^{-}\left(x\right)= x^4\!+\!(-\!4t^2-\!5)x^3\!+\!(4t^4\!+\!2t^3\!+\!5t^2\!-\!2t\!+\!6)x^2\!+\!(-t^6\!-\!2t^5\!\!+\!4t^3\!-\!3t^2\!-\!1)x.
\end{align*}
Using the notation from Section 4, $f=\chi_{\BPhi}^{-}$ has the coefficient sequence of the form $\underline{a}=(a_4,a_3,a_2,a_1)$, where
\begin{align*}
a_4(t) &=1,\\
a_3(t) &=-\!4t^2-\!5,\\
a_2(t) &=4t^4\!+\!2t^3\!+\!5t^2\!-\!2t\!+\!6,\\
a_1(t) &=-t^6\!-\!2t^5\!+\!4t^3\!-\!3t^2\!-\!1.
\end{align*}
Note that a sign variation formula $\varphi(t)=\varphi(t)_{\sigma}^f$ which is satisfied, must have a positive component $a_4(t)>0$ and a negative component $a_3(t)<0$. If in addition $\lambda\left(\sigma\right)>(n-1)^{2}=1$, then $\varphi(t)=\varphi(t)_{\sigma}^f$
is one of the following:
\begin{align*}
\varphi_{1}(t) & =1\!>\!0\;\wedge\; a_{3}(t)\!<\!0\;\wedge\; a_{2}(t)\!>\!0\;\wedge\; a_{1}(t)\!<\!0,\\
\varphi_{2}(t) & =1\!>\!0\;\wedge\; a_{3}(t)\!<\!0\;\wedge\; a_{2}(t)\!<\!0\;\wedge\; a_{1}(t)\!>\!0,\\
\varphi_{3}(t) & =1\!>\!0\;\wedge\; a_{3}(t)\!<\!0\;\wedge\; a_{2}(t)\!>\!0\;\wedge\; a_{1}(t)\!>\!0.
\end{align*}
We apply Algorithm 7.2 for $\varphi_{1}(t)$.
\begin{description}
\item [(S1)] The normal form of $\varphi_{1}(t)$ is as follows:
\[
1\!>\!0\;\wedge\; -a_{3}(t)\!>\!0\;\wedge\; a_{2}(t)\!>\!0\;\wedge\; -a_{1}(t)\!>\!0
\]
\item [(S2)] Since $-a_3(t)>0$, for any $t\in\RR$, it is enough to determine whether the condition $$(N(a_{2}(t),\textnormal{\textbf{1}})=0\wedge a_{2}(0)>0)\;\wedge\;(N(-a_{1}(t),\textnormal{\textbf{1}})=0\wedge -a_{1}(0)>0)$$ holds. We get the following calculations:
\begin{itemize}
\item The case of $N(a_{2}(t),\textnormal{\textbf{1}})=0\wedge a_{2}(0)>0$:
\begin{itemize}
\item Obviously $a_{2}\left(0\right)>0$. Set $h_{0}(t)=a_{2}(t)$, $h_{1}(t)=h_{0}'(t)$ and construct the Sturm sequence:
\begin{align*}
h_{0}(t) & =4t^4\!+\!2t^3\!+\!5t^2\!-\!2t\!+\!6\\
h_{1}(t) & =16t^{3}+6t^2+10t-2\\
h_{2}(t) & =-\frac{37}{16}t^{2}+\frac{29}{16}t-\frac{97}{16}\\
h_{3}(t) & =\frac{23840}{1369}t+\frac{69280}{1369}\\
h_{4}(t) & =\frac{10961583}{355216}.
\end{align*}
\item Since
\begin{align*}
\lambda\left(\sigma_{-}\left(h_{0},\ldots,h_{4}\right)\right) & =\lambda\left(+---+\right)=2\\
\lambda\left(\sigma_{+}\left(h_{0},\ldots,h_{4}\right)\right) & =\lambda\left(++-++\right)=2,
\end{align*}
we have $N\left(a_{2},\textnormal{\textbf{1}}\right)=0$. Thus $N(a_{2}(t),\textnormal{\textbf{1}})=0\wedge a_{2}(0)>0$ holds.
\end{itemize}
\item The case of $N(-a_{1}(t),\textnormal{\textbf{1}})=0\wedge -a_{1}(0)>0$:
\begin{itemize}
\item Obviously $-a_{1}\left(0\right)>0$. Set $h_{0}(t)=-a_{1}(t)$, $h_{1}(t)=h_{0}'(t)$ and construct the Sturm sequence:
\begin{align*}
h_{0}(t) & =t^6\!+\!2t^5\!-\!4t^3\!+\!3t^2\!+\!1\\
h_{1}(t) & =6t^{5}+10t^{4}-12t^2+6t\\
h_{2}(t) & =\frac{5}{9}t^4+2t^3- \frac{8}{3}t^2+\frac{1}{3}t-1\\
h_{3}(t)& =-\frac{1764}{25}t^3+\frac{1782}{25}t^2-\frac{594}{25}t+\frac{522}{25}\\
h_{4}(t) & =\frac{7675}{28812}t^2+\frac{31525}{86436}t+\frac{2325}{9604}\\
h_{5}(t) & = \frac{446509168}{2356225}t+\frac{62061048}{471245}\\
h_{6}(t) & = -\frac{86947058725}{741396579952}
\end{align*}
\item Since
\begin{align*}
\lambda\left(\sigma_{-}\left(h_{0},\ldots,h_{4}\right)\right) & =\lambda\left(+-+++--\right)=3\\
\lambda\left(\sigma_{+}\left(h_{0},\ldots,h_{4}\right)\right) & =\lambda\left(+++-++-\right)=3,
\end{align*}
we have $N\left(-a_{1},\textnormal{\textbf{1}}\right)=0$. Thus $N(-a_{1}(t),\textnormal{\textbf{1}})=0\wedge -a_{1}(0)>0$ holds.
\end{itemize}
\end{itemize}
\end{description}
Therefore Algorithm 7.2 outputs \textit{$\BPhi$ is globally nonpositive}.
\end{ex}

\begin{remark} Observe that in the example above we have \[
\Phi_{t}\left( X_0\right)=\left[\begin{array}{cc}
-2t^2-1 & 4t+1\\
4t+1 & -7t^{2}+2t-8
\end{array}\right]
\] where \[
X_0=\left[\begin{array}{cc}
1 & 1\\
1 & 2
\end{array}\right]
\] is a semipositive matrix. It can be shown that \[
\frac{1}{2}\left(-9t^{2}+2t-9-\sqrt{25t^{4}-20t^{3}+138t^{2}+4t+53}\right)
\] is a eigenvalue of $\Phi_{t}\left(X_{0}\right)$, for any $t\in\RR$. Since the number is always negative, this also proves that the one-parameter family $\BPhi$ is globally nonpositive. Recall however that these kind of arguments are possible only in the case of matrices of small sizes. Indeed, it follows from the Galois theory that we are not able to calculate directly the eigenvalues of $n\times n$ matrices with $n\geq 5$. Nevertheless, our algorithms can be applied in such situations as well.
\end{remark}

\end{document}